\documentclass[12pt]{article}
\usepackage[utf8]{inputenc}
\usepackage{graphicx}
\usepackage{amsmath}
\usepackage{amsfonts}
\usepackage{amssymb}
\usepackage{float}
\usepackage{subcaption}
\usepackage{euscript} 
\usepackage{ mathrsfs }
\usepackage{amsthm}
\usepackage{hyperref}
\usepackage{authblk}
\usepackage{xcolor, soul}
\usepackage{cleveref}

\usepackage{setspace}
\newtheorem{theorem}{Theorem}[section]
\newtheorem{lemma}{Lemma}[section]

\theoremstyle{remark}

\newcommand{\be}{\begin{equation}}
\newcommand{\ee}{\end{equation}}
\newcommand{\beq}{\begin{eqnarray}}
\newcommand{\eeq}{\end{eqnarray}}
\newcommand{\beqn}{\begin{eqnarray*}}
\newcommand{\eeqn}{\end{eqnarray*}}
\def\bes{\begin{equation*}}
\def\ees{\end{equation*}}

\title{A greedy algorithm for habit formation under multiplicative utility}
\author[1]{S. Kirusheva\footnote{skirusheva@gmail.com}}
\author[2]{T. S. Salisbury\footnote{salt@yorku.ca}}
\affil[1,2]{York University}
\date{}
\begin{document}
\maketitle
\begin{abstract}
We consider the problem of optimizing lifetime consumption under a habit formation model, both with and without an exogenous pension.  Unlike much of the existing literature, we apply a power utility to the ratio of consumption to habit, rather than to their difference. The martingale/duality method becomes intractable in this setting, so we develop a greedy version of this method that is solvable using Monte Carlo simulation. We investigate the behaviour of the greedy solution, and explore what parameter values make the greedy solution a good approximation to the optimal one.
\end{abstract}
\section{Introduction}
\subsection{Overview}
We consider a retiree who wants to optimize their consumption in retirement. They evaluate the utility of current consumption in relation to past consumption or \emph{habit}, in other words, their utility function depends not just on current consumption, but also on habit -- an exponentially weighted average of past consumption. 

There is a large habit formation literature, but almost all of this assumes an additive relationship between habit $H_t$ and current consumption $C_t$. Often this means applying a power law utility $u(x)=\frac{x^{1-\gamma}}{1-\gamma}$ to the difference of the two, ie $u(C_t-H_t)$. This choice leads to some elegant simplifications, but also some unrealistic consequences. For example, if risk aversion $\gamma$ is $>1$, then we must always have $C_t>H_t$. Therefore consumption always rises over time. 

We therefore consider an alternate multiplicative form for utility, due to \cite{Rogers}, in which the power law utility function is applied to the ratio $\frac{C_t}{H_t}$. This seems to us to be economically more natural. Also, in the absence of the simplifications arising in the special case of additive utility, solving this problem brings us closer to being able to treat more general and therefore more realistic formulations. 

\cite{KHS} took steps in this direction, using the classic value function methodology and PDE's. But the more realistically one formulates a problem, the higher the dimensionality of the PDE becomes. This issue already appeared in \cite{KHS} where the presence of exogenous pension income removed the possibility of dimension reduction via scale invariance. To introduce further realism could easily tip such an approach over the edge to infeasibility. 

We therefore wanted to explore alternative approaches to this problem, that are less affected by the curse of dimensionality. A natural candidate is the martingale or duality method, which is the subject of this paper. And which, in our computations, will rely on Monte Carlo simulation rather than PDE's. 

Duality works well for additive habit formation, but it runs into significant problems in the multiplicative case. We therefore formulate a very similar problem, for which duality does work nicely, and for which it is reasonable to expect that the two solutions would be close. In other words, we propose an approximate solution to the original optimization problem. We call our solution the \emph{Greedy Optimum}. We will analyze its properties and will show that there are non-trivial choices of the model parameters for which it provides a good approximation. There are other choices for which it performs poorly, and exploring which is which is a main goal of this research.

\subsection{Literature review and agenda}
Many articles have been written about habit formation and optimal consumption. Almost all of these involve additive utility, as discussed above. This allows clever simplifications, starting with the pioneering work of \cite{C} and \cite{DZ}. The literature has explored many alternative choices for portfolio dynamics, labour income, utility, etc. We will focus particularly on papers that use the martingale/duality approach to these problems. 

Several sources focus (as we do)  particularly on retirement, mortality, or lifecycle planning. For example \cite{B} does so in a habit formation model incorporating stochastic wages and labor supply flexibility. \cite{LLYW} allows non-exponential discounting and the purchase of insurance. See also \cite{JP} or \cite{Rei}

Other works focus more on financial questions. For example, \cite{KLSX} allows an incomplete market containing a bond and multiple stocks. \cite{LWZ} decomposes consumption into both habitual and non-habitual components. See also \cite{CA}, \cite{CH}, \cite{DK}, \cite{NE}, \cite{HKW}, \cite{KLSS}, \cite{KLS}, \cite{CM}, \cite{VBG}, \cite{Yu} and \cite{XYu}. For some very recent work, see \cite{ABY} or \cite{HHJ}.

This paper is organized as follows.  Section \ref{COP_HFM} formulates the optimal and greedy approaches, and characterizes the greedy solutions analytically. Section \ref{S_NR} gives our numerical results, showing the behaviour of the greedy solutions. In Section \ref{comp_P_M} we consider the accuracy of the Greedy approximation, comparing with exact results from  \cite{KHS} obtained using the value function/PDE approach. Some simulations of how wealth and consumption actually evolve over time are given in Section \ref{WD}.  

\noindent{\bf Acknowledgement:} This work forms part of the Ph.D. Thesis of Kirusheva, at York University. Salisbury's research is supported, in part, by NSERC. This paper is dedicated to the memory of Tom Hurd, a friend and valued colleague of Salisbury for 34 years. You are sorely missed.
\section{Habit formation: global and greedy optima}
\label{COP_HFM}
In this section we will formulate the global optimization problem and its greedy counterpart, and we will show how to solve the latter using Monte Carlo methods. 
\subsection{Notation}
We start with a stock price process $S_t$ following a geometric Brownian motion
\bes
dS_t=\mu S_t\,dt + \sigma S_t\,dW_t
\ees
where $W_t$ is a Brownian motion under $P$. We will always work relative to the filtration $\EuScript{F}_t$ of $B_t$. Our controls will be consumption $C_t$ and asset allocation $\theta_t$, so our wealth dynamics are that
\beq
\label{wealth_dynamics}
dX_t=[\theta_t(\mu-r)+r]X_t \,dt+\theta_t \sigma X_t  dW_t-C_t \,dt +\pi\,dt.
\eeq
Here $r$ is the risk-free rate and $\pi$ is the rate of exogenous pension income (which may $=0$). 

Let 
\bes
u(c)=\frac{c^{1-\gamma}}{1-\gamma}
\ees
be the power law or CRRA utility associated to a risk-aversion $\gamma>0$, $\gamma\neq 1$. It will be convenient to consider a general habit stream $H_t$, and to formulate the lifetime utility of the pair $(C_t, H_t)$ as 
\be
\EuScript{U}(C,H)=E\left[\int_0^\infty e^{-\rho s}{}_sp_x u\left(\frac{C_s}{H_s}\right)\, ds\right],
\label{f1_g_vs_o}
\ee
where $\rho$ is the subjective discount rate, ${}_sp_x=e^{-\int_0^s\lambda_{x+s}}\,ds$ is the probability of an individual of age $x$ surviving $s$ years, and $\lambda_y$ is the hazard rate at age $y$. In computations we will typically assume that the hazard rate takes the Gompertz form, ie $\lambda_y=\frac{1}{b}e^{\frac{y-m}{b}}$, where $m$ is the modal value of life (see p47, \cite{M}), and $b$ is a dispersion coefficient. In this case ${}_sp_x=e^{-e^{\frac{x-m}{b}}(e^{\frac{s}{b}}-1)}$.

The particular habit stream generated by a consumption stream $C_t$ and an initial habit $\bar c$ will be denoted $\mathscr{H}(C, \bar c)$. In other words, $H_\cdot=\mathscr{H}(C, \bar c)$ means that
\begin{align}
dH_t&=\eta (C_t-H_t)\,dt\\
H_0&=\bar c
\label{3eq_1}
\end{align}
Here $\eta$ is a parameter of the model, that represents how fast the client's habit reacts to changes in consumption.

Let $\kappa=\frac{\mu-r}{\sigma}$, so that $\zeta_t=e^{-rt}e^{-\kappa W_t-\frac{\kappa^2}{2}t }$ represents the state price density. In other words, $\zeta_t=e^{-rt}\xi_t$, where $\xi_t$ is the Randon Nikodym derivative of the risk neutral measure $Q$ with respect to $P$. We will also need the notation $\tilde \zeta^t_s=e^{-r(s-t)}e^{-\kappa \tilde W_{s-t}-\frac{\kappa^2}{2}(s-t) }$, where $\tilde W_q=W_{t+q}-W_t$. In particular, $\zeta_s=\zeta_t\tilde\zeta^t_s$ for $s>t$, and the conditional distribution of $\tilde\zeta^t_{s}$ given $\EuScript{F}_t$ is the same as the unconditioned distribution of $\zeta_{s-t}$. 

Let $v$ denote initial wealth, and assume now that $\pi=0$. At various points we will impose one or another of the following conditions. The \emph{budget constraint} for a consumption stream $C_t\ge 0$ is the following:
\be
E\left[\int_0^\infty \zeta_s C_s\,ds\right]\le v.
\label{BudgetConstraint}
\ee
The \emph{exact budget constraint} is likewise, but with equality. 
Alternatively, a pair $(\theta_t, C_t)$ is said to be \emph{admissible} if each is adapted, $X_0=v$, and for every $t$ both $C_t\ge 0$ and $X_t\ge 0$. The connection between these is as follows: 
\begin{lemma} Let $C_t\ge 0$ be adapted. It satisfies the budget constraint $\Leftrightarrow$ there is a $\theta_t$ for which $(\theta_t,C_t)$ is admissible. 
\label{admissibility}
\end{lemma}
\begin{proof} This is a familiar result, but we will sketch the argument, for later use. First assume that $(\theta_t,C_t)$ is admissible. Set $M_t=\zeta_tX_t+\int_0^t\zeta_sC_s\,ds$. An application of Ito's formula shows that 
\be
dM_t=\zeta_tX_t(\sigma\theta_t-\kappa)\,dW_t,
\label{portfolio}
\ee
so $M_t$ is a martingale with $M_0=v$. Since $X_t\ge 0$, we may take expectations to obtain \eqref{BudgetConstraint}. Conversely, assume \eqref{BudgetConstraint} and let $\delta=v-E[\int_0^\infty \zeta_s C_s\,ds]\ge 0$. Define $X_t$ by 
\bes
\zeta_tX_t=\delta+E[\int_t^\infty\zeta_s C_s\,ds\mid\EuScript{F}_t].
\ees
Note that $X_t\ge 0$, since $C_t\ge 0$. Then $M_t=\zeta_tX_t+\int_0^t\zeta_sC_s\,ds$ is easily seen to be a martingale, and the martingale representation theorem may now be used to find $\theta_t$ so that \eqref{portfolio} holds, at least while $X_t>0$. And if ever $X_t=0$ then by definition, $C_s=0$ and $X_s=0$ for a.e. $s>t$.

It is now simple to undo the use of Ito's formula to see that $(\theta_t,C_t)$ is admissible.
\end{proof}
For later use, observe that if the exact budget constraint holds, the argument shows that we may obtain a portfolio process via
\be
\zeta_tX_t=E[\int_t^\infty\zeta_s C_s\,ds\mid\EuScript{F}_t].
\label{portfolioprocess}
\ee

\subsection{Optimal solution vs. greedy heuristic solution: no pension case}
\label{OP_vs_GH}
In this section we assume that the exogenous pension $\pi=0$. 
We may now formulate the global optimization problem, and its greedy counterpart. Recall that a greedy algorithm is one that optimizes some local quantity, without taking into account how that choice may affect other quantities. Sometimes a greedy algorithm may be globally optimal, but often it is not. It is however usually simpler to compute. 
\medskip

\noindent\emph{Globally optimal formulation:}

In this version, we wish to optimize $\EuScript{U}(C,H)$ subject to the following constraints:
\begin{itemize}
\item $(\theta_t,C_t)$ is admissible,
\item  $H_\cdot=\mathscr{H}(C, \bar c)$
\end{itemize}

Lemma \ref{admissibility} can be used to show that this is equivalent to maximizing  $\EuScript{U}(C,H)$ over adapted $C_t\ge 0$ satisfying the budget constraint and for which $H_\cdot=\mathscr{H}(C, \bar c)$. 
\medskip

\noindent\emph{Greedy formulation:}

We seek $C^*_t$ satisfying the exact budget constraint, such that
if we set $H_\cdot=\mathscr{H}(C^*, \bar c)$, then $C=C^*$ maximizes $\EuScript{U}(C,H)$ over adapted $C_t\ge 0$ satisfying the budget constraint \eqref{BudgetConstraint}. 
\medskip

In other words, modifying this $C^*_t$ will not improve utility directly, but might do so indirectly, through modifying habit. 

We say that such a consumption stream $C^*_t$ is \emph{locally optimal}, or a \emph{greedy optimum}. For very small values of $\eta$, habit should not be relevant, so both versions should be well approximated by the classic Merton problem. But we can hope that for slightly larger values of $\eta$ the greedy solution will still provide a good approximation to the global optimum, even when the Merton approximation ceases to be good.

We will obtain the greedy optimum numerically, by solving certain equations. The following verification theorem exhibits these equations, and shows that their solution will indeed yield a greedy optimum. 

We will see that for any $\alpha$, (a) and (b) below can be solved, following which the expectation in \eqref{BudgetConstraint} can be computed by numerical integration combined with Monte Carlo simulation. This gives a value that depends on $\alpha$, but then $\alpha$ can be adjusted, and the process repeated, until a solution is found for which the exact budget constraint is satisfied. 

\begin{theorem}
\label{lem}
Suppose there exists an adapted consumption stream $C^*_t\ge 0$ and a Lagrange multiplier $\alpha>0$ such that the following conditions hold: 
\begin{enumerate}
\item $C^*_t=H_t^{1-\frac{1}{\gamma}}\left(\alpha e^{\rho t}{}_tp_x^{-1}\zeta_t\right)^{-\frac{1}{\gamma}}$,
\item $H_\cdot=\mathscr{H}(C^*, \bar c)$, and 
\item $C^*_t$ satisfies the exact budget constraint.
\end{enumerate}
Then this consumption stream $C^*_t $ is a greedy optimum.
\end{theorem}
\begin{proof}
Let $C_s$ be any adapted consumption stream which satisfies the budget constraint (\ref{BudgetConstraint}) 
with $C_s\ge0 \;\forall s.$ First, consider the problem of finding $C_t$ to maximize the following quantity, with $H_\cdot$ fixed:
\be
\int_0^\infty e^{-\rho s}{}_sp_xu\left(\frac{C_s}{H_s}\right)\,ds \quad {\rm s.t.}\quad \int_0^\infty \zeta_s C_s\, ds=v.
\label{pathwise_equation}
\ee
We use the method of Lagrange multipliers. Fix $\alpha>0$. Then (a) and simple calculus shows that
\begin{multline}
\int_0^\infty e^{-\rho s}{}_sp_xu\left(\frac{C_s}{H_s}\right)\,ds-\alpha\left(\int_0^\infty \zeta_s C_s\, ds-v\right)\\
\le \int_0^\infty e^{-\rho s}{}_sp_xu\left(\frac{C^*_s}{H_s}\right)\,ds-\alpha\left(\int_0^\infty \zeta_s C^*_s\, ds-v\right).
\end{multline}
Taking expectations, we see that
\begin{multline}
\EuScript{U}(C,H)-\alpha\left(E\left[\int_0^\infty \zeta_s C_s\,ds\right]-v\right)  \\
\qquad\le  \EuScript{U}(C^*,H)-\alpha\left(E\left[\int_0^\infty \zeta_s C^*_s\,ds\right]-v\right)
=\EuScript{U}(C^*,H)
\end{multline}
because $C^*$ satisfies the exact budget constraint. Therefore, if $C$  satisfies \eqref{BudgetConstraint}, we  conclude that  $\EuScript{U}(C,H)\le \EuScript{U}(C^*,H)$, as required.
\end{proof}

This argument also shows why the global optimum is not amenable to solution via straightforward duality methods. Namely that when $H_t$ is not held fixed when applying Lagrange multipliers, problematic terms such as $\frac{\partial H_t}{\partial C_s}$ will arise for $t>s$ when differentiating the first expression in \eqref{pathwise_equation}.

An interesting observation is that the above allows us solve for $C^*_t$ and $H_t$ (in terms of the observable state-price-density $\zeta_t$), without needing to also find the account value process $X_t$ or the asset allocation process $\theta_t$. As we will see below, these latter quantities are in fact both computable, though at the expense of significantly more effort than goes into the computations of $C^*_t$ and $H_t$. 

As stated above, if we fixe $\alpha$ then (a) and (b) can be solved. In fact, there is an explicit formula for $H_t$, which we now proceed to describe, from which (a) gives us $c^*_t$. It turns out that this will not be feasible once we introduce pension income, which is one reason why the formulas of this section are somewhat simpler than those we will derive later. 

Combining (a) and (b) of Theorem \ref{lem} gives the following differential equation:
\bes
H'_s + \eta H_s-\eta(\alpha\zeta_se^{\rho s}{}_sp_x^{-1})^{-\frac{1}{\gamma}}H_s^k=0
\ees
where $k=1-\frac{1}{\gamma}$. This is a Bernoulli equation, which can be solved with a standard substitution, to give
\bes
H_s=e^{-\eta s}\left(\frac{\eta}{\gamma} \int_0^s (\alpha\zeta_q e^{(\rho-\eta) q}{}_qp_x^{-1})^{-\frac{1}{\gamma}} \,dq+\bar c^{\frac{1}{\gamma}}\right)^\gamma.
\ees
More generally, given $H_t$ we can find $H_s$ for $s>t$ as
\beq
H_s=e^{-\eta s}\left(\frac{\eta}{\gamma} \int_t^s (\alpha\zeta_q e^{(\rho-\eta) q}{}_qp_x^{-1})^{-\frac{1}{\gamma}}\,dq+H_t^{\frac{1}{\gamma}}e^{\frac{\eta t}{\gamma}}\right)^\gamma.
\label{4eq_18}
\eeq
This will imply the following expressions for $X_t$ and $\theta_t$.
\begin{theorem}
\label{lemma_wo_pension}
Consider a greedy optimum $(C^*_t,H_t)$ as given in Theorem \ref{lem}. 
Then the wealth process $X_t$ has the form
\be
X_t=\zeta_t^{-1}F(t,Z_t) \quad {\it where} \; Z_t=\zeta_t H_t
\label{4eq_29}
\ee
and
\begin{multline}
F(t,z)=\alpha^{-\frac{1}{\gamma}}E\left[\int_t^\infty (\tilde \zeta^t_s)^{1-\frac{1}{\gamma}}e^{\frac{-\rho s}{\gamma}}{}_sp_x^{\frac{1}{\gamma}}\right. \\
\left. \cdot\left(\frac{\eta}{\gamma}\int_t^s\left(\alpha\tilde\zeta^t_qe^{\rho q}e^{\eta(s-q)}{}_qp_x\right)^{-\frac{1}{\gamma}}\,dq+z^{\frac{1}{\gamma}}e^{-\frac{\eta (s- t)}{\gamma}}\right)^{\gamma-1}\,ds \right]
\label{4eq_28}
\end{multline}
Moreover the portfolio process 
\beq
\theta_t=\frac{\kappa}{\sigma}\left( 1-\frac{Z_t F_z(t,Z_t)}{F(t,Z_t)}\right).
\label{thetaformula}
\eeq
makes $(\theta_t,C^*_t)$ admissible.
\end{theorem}
\begin{proof} Since $\zeta_s=\zeta_t\tilde\zeta^t_s$, 
\begin{align*}
&\int_t^\infty \zeta_s C^*_s\,ds 
= \int_t^\infty \zeta_s H_s^{1-\frac{1}{\gamma}}\left(\alpha e^{\rho s}{}_sp_x^{-1}\zeta_s\right)^{-\frac{1}{\gamma}}\,ds\\
&\qquad= \int_t^\infty \zeta_s \left[e^{-\eta s}\left(\frac{\eta}{\gamma} \int_t^s (\alpha\zeta_q e^{(\rho-\eta) q}{}_qp_x^{-1})^{-\frac{1}{\gamma}} \,dq+H_t^{\frac{1}{\gamma}}e^{\frac{\eta t}{\gamma}}\right)^\gamma\right]^{1-\frac{1}{\gamma}} \left(\alpha e^{\rho s}{}_sp_x^{-1}\zeta_s\right)^{-\frac{1}{\gamma}}\,ds\\
&\qquad= \alpha^{-\frac{1}{\gamma}}\int_t^\infty \tilde (\zeta^t_s)^{1-\frac{1}{\gamma}} e^{-\frac{\rho s}{\gamma}}{}_sp_x^{\frac{1}{\gamma}} \left[\frac{\eta}{\gamma} \int_t^s (\alpha\tilde \zeta^t_q e^{\rho q}e^{\eta(s-q)}{}_qp_x)^{-\frac{1}{\gamma}} \,dq+\zeta_t^{\frac{1}{\gamma}}H_t^{\frac{1}{\gamma}}e^{-\frac{\eta (s-t)}{\gamma}}\right]^{\gamma-1} \,ds.
\end{align*}
But $\tilde\zeta^t_s$ is independent of $\EuScript{F}_t$, so \eqref{4eq_29} follows by \eqref{portfolioprocess}.

Recall from Lemma \ref{admissibility} that $M_t=\zeta_t X_t+\int_0^t \zeta_s C_s^*\, ds$ is a martingale. By \eqref{portfolio} we have that
\bes
dM_t=F(t,Z_t)(\sigma\theta_t-\kappa)\,dW_t.
\ees
But we may also substitute $F(t,Z_t)$ for $\zeta_tX_t$ and apply Ito's formula to obtain that
\bes
dM_t=-\kappa Z_tF_z(t,Z_t)\,dW_t.
\ees
Formula \eqref{thetaformula} now follows. 
\end{proof}
If we check the last formula for the case $\eta=0$  we will get $\theta_t=\frac{\kappa}{\sigma \gamma},$ which coincides with the solution of the Merton problem.

In any case, if we simulate a path for the stock $S_t$ (or equivalently, a path for $\zeta_t$), we can then compute the corresponding paths for $C_t$, $H_t$, $X_t$, and $\theta_t$. Of course, finding $C_t$ and $H_t$ is less expensive than finding $X_t$ and $\theta_t$, because once the Monte Carlo runs to find $\alpha$ are finished, we will have the entire paths of the former,  whereas the latter will require additional runs for each $t$ considered, in order to compute $F$ and $F_z$.

As a check on our results, we compared $X_t$, computed as above, with $X_t$ computed using an Euler-Maruyama discretization. Of course, this requires $\theta_t$, so that quantity was still computed via Monte Carlo simulation. Good agreement was obtained.

\subsection{Adding pension}
\label{HFMwith_pension}
We now want to explore the same ideas, in the case of a positive exogenous pension $\pi>0$. If we had taken a PDE approach, this would destroy scale invariance, meaning that a dimension reduction would no longer be available, rendering this case significantly harder to solve. With the martingale approach, we will see that some aspects become a bit more complicated, but not dramatically so. 

We will distinguish between consumption from wealth $C_t^{w}$ and total consumption $C_t=C_t^{w}+\pi$. Because we focus on the martingale approach, and so will apply Lemma \ref{admissibility} to $C^w_t$, we will need to impose a constraint that $C^w_t\ge 0$. Or equivalently, $C_t\ge\pi$. Whether this is economically natural or not will, of course, depend on the size of $\pi$. Admissibility and the budget constraint are as before, except now in relation to $C^w_t$. 

This means our two problems will be formulated as follows.
\medskip

\noindent\emph{Globally optimal formulation:}

\noindent Choose $\theta_t$ and $C_t$ to optimize $\EuScript{U}(C,H)$ subject to the following constraints:
\begin{itemize}
\item $C_t\ge\pi$ for every $t$,
\item $(\theta_t,C^w_t)$ is admissible,
\item  $H_\cdot=\mathscr{H}(C, \bar c)$
\end{itemize}
\medskip

\noindent\emph{Greedy formulation:}

\noindent We seek $C^*_t$ such that the following hold:
\begin{itemize}
\item $C^*_t\ge \pi$ for every $t$,
\item $C^{w,*}_t=C^*_t-\pi$ satisfies the exact budget constraint, and 
\item If we set $H_\cdot=\mathscr{H}(C^*, \bar c)$, then $C=C^*$ maximizes $\EuScript{U}(C,H)$ over adapted $C_t\ge \pi$ such that $C^w_t$ satisfies the budget constraint \eqref{BudgetConstraint}. 
\end{itemize}
\medskip

\begin{theorem}
\label{theor}
Suppose there exists an adapted consumption stream $C^*_t$ and a Lagrange multiplier $\alpha>0$ such that the following conditions hold: 
\begin{enumerate}
\item $C^*_t=\pi\lor\left(H_t^{1-\frac{1}{\gamma}}\left(\alpha e^{\rho t}{}_tp_x^{-1}\zeta_t\right)^{-\frac{1}{\gamma}} \right)$,
\item $H_\cdot=\mathscr{H}(C^*, \bar c)$, and 
\item $C^{w,*}_t$ satisfies the exact budget constraint.
\end{enumerate}
Then this consumption stream $C^*_t $ is a greedy optimum.
\end{theorem}
\begin{proof}
Our Lagrangian now becomes
\bes
\int_0^\infty e^{-\rho s}{}_sp_xu\left(\frac{C^w_s+\pi}{H_s}\right)\,ds-\alpha\left(\int_0^\infty \zeta_s C^w_s\, ds-v\right).
\ees
and maximizing it leads to 
\be
C^w_s=\left(H_s^{1-\frac{1}{\gamma}}(\alpha e^{\rho s}{}_sp_x^{-1}\zeta_s)^{-\frac{1}{\gamma}}-\pi\right)_+.
\label{Cwformula}
\ee
Note that the first order condition is only binding when the bracketed expression is $>0$. The remainder of the argument is as in Theorem \ref{lem}.
\end{proof}

The $\lor$ in (a) complicates the analytic solution of (a) and (b), so in this case we will opt to solve them numerically using the Euler--Maruyma method. As before, we then iterate to find the $\alpha$ that makes (c) hold, using Monte Carlo simulation to compute the necessary expectations.

With pension, the expressions for $X_t$ and $\theta_t$ are somewhat more complicated, now depending on both $\zeta_t$ and $H_t$ rather than on their product.  They are still amenable to calculation via Monte Carlo simulation.
\begin{theorem}
\label{lemma_with_pension}
Consider a greedy optimum $(C^*_t,H_t)$ as given in Theorem \ref{theor}. 
Then the wealth process $X_t$ has the form
\be
X_t=G(t,\zeta_t,H_t) 
\label{4eq_29}
\ee
where
\be
G(t,y,h)=E\left[\int_t^\infty \tilde \zeta^t_s\left(H_s^{1-\frac{1}{\gamma}}(\alpha e^{\rho s}{}_sp_x^{-1}y\tilde\zeta^t_s)^{-\frac{1}{\gamma}}-\pi\right)_+ \,ds\mid H_t=h\right]
\label{3eq_60}
\ee
Moreover the portfolio process 
\bes
\theta_t=-\frac{\kappa\zeta_tG_y(t,\zeta_t,H_t)}{\sigma G(t,\zeta_t,H_t)}
\label{thetaformula}
\ees
makes $(\theta_t,C^{w,*}_t)$ admissible.
\end{theorem}
\begin{proof}
From \eqref{portfolioprocess} we have that
\bes
\zeta_t X_t=E\left[\int_t^\infty \zeta_sC^{w,*}_s\,ds\mid\EuScript{F}_t\right].
\ees
Now substitute \eqref{Cwformula} and use that $\zeta_s=\zeta_t\tilde\zeta^t_s$ to obtain that
\bes
\zeta_tX_t=E\left[\int_t^\infty \zeta_t\tilde \zeta^t_s\left(H_s^{1-\frac{1}{\gamma}}(\alpha e^{\rho s}{}_sp_x^{-1}\zeta_t\tilde\zeta^t_s)^{-\frac{1}{\gamma}}-\pi\right)_+ \,ds\mid \EuScript{F}_t\right].
\label{3eq_60}
\ees
\eqref{3eq_60} now follows, using that $\tilde \zeta^t_s$ is independent of $\EuScript{F}_t$, which means that the only relevant information in $\EuScript{F}_t$ is the initial condition $H_t$ for the ODE giving $H_s$.

Computing $dM_t$ as before we obtain $-\kappa\zeta-t[G+\zeta_tG_y]=\zeta_t(\sigma\theta_t-\kappa)G$, from which \eqref{thetaformula} follows.
\end{proof}

\section{Behaviour of the greedy solution}
\label{S_NR}
Numerical results for the greedy solution will be presented in two parts, first without exogenous pension income ($\pi=0$), and then with $\pi>0$. 

Unless specifically mentioned otherwise, we will use the following parameter values: {\bf  risk-free rate} $r=0.02,$ {\bf volatility} $\sigma=0.16,$ {\bf drift} $\mu=0.08,$  {\bf subjective discount rate} $\rho=0.02$ and {\bf risk aversion} parameter $ \gamma=3$. {\bf Gompertz parameters} will be {\bf age} $x=65$, $m=89.335$, and $b=9.5$ for consistency with \cite{KHS}.

We will vary the parameter $\eta$ that reflects how fast the habit formation model reacts to the client's consumption choices, i.e. the  smoothing factor $\eta$. In our calculations, we will take $\eta=0.01,\; 0.1$ and $1$. 

The dynamics are such that $C^*_t$ and $\theta_t$ are both functions of time $t$, habit $H_t$, and wealth $X_t$. In the plots below, we will examine the dependence on $X_t$, by fixing $H_t=1$ and showing differently coloured curves for five different times. Specifically, $t=0$(green), $10$(blue), $20$(red), $30$(black), and $40$(magenta). 

Note that when $\pi=0$, the $H_t=1$ plots tell the whole story, as the full problem could be reformulated in terms of $C^*_t/H_t$ and $X_t/H_t$. See \cite{Rogers} or \cite{KHS}.
\subsection{Without pension income}
\label{NR_wo_pension}
\begin{figure}[H]\centering
\begin{minipage}[b]{\textwidth} 
\includegraphics[scale=.34]{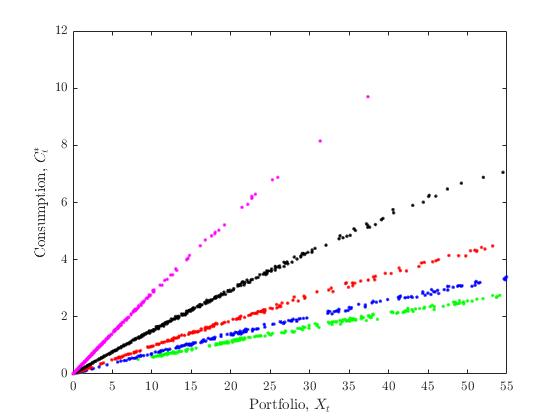}
\includegraphics[scale=.34]{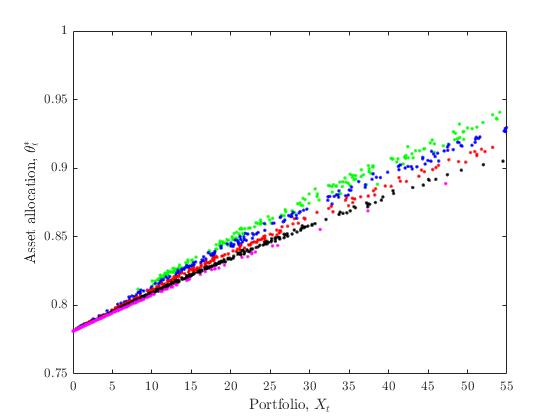}
\end{minipage}
\caption{Greedy solution with $\eta=0.01$, $\pi=0$ and habit $H_t=1,$ for multiple time moments.}
\label{f_1}
\end{figure}
Figure \ref{f_1} corresponds to the smallest value of the smoothing factor $\eta=0.01.$ The client's habit does not adapt to consumption rapidly, therefore this numerical solution will be closest to that of the the case $\eta=0$, which is the classic Merton problem (but with ${}_tp_x$ included in the discount factor). Indeed, the plots on the left are nearly linear, as they would be under Merton. Note the scale of the plots on the right -- there is relatively little change in $\theta_t$, which is reminiscent of the constant allocation we would see under Merton. 

Of course, consumption rises with time, to reflect the shorter horizon till death over which to consume.

%
\begin{figure}[H]\centering
\begin{minipage}[b]{\textwidth} 
\includegraphics[scale=.34]{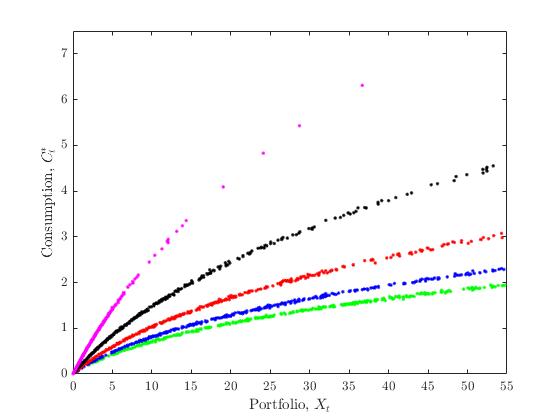}
\includegraphics[scale=.34]{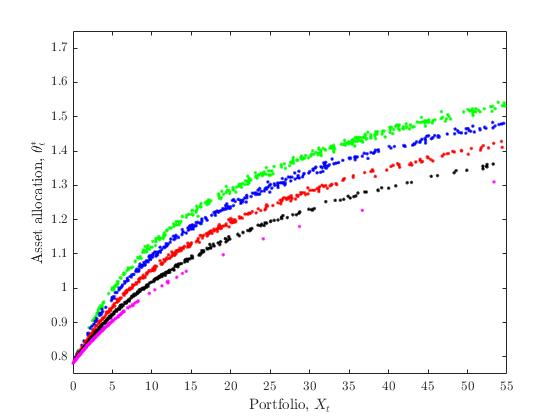}
\end{minipage}
\caption{Greedy solution with $\eta=0.1$, $\pi=0$ and habit $H_t=1,$ for multiple time moments.}\label{f3_2}
\end{figure}
In Figure \ref{f3_2} we take $\eta=0.1$, so the impact of habit should be greater than before, since habit will adapt quicker to consumption. affect results more. Consumption is no longer linearly related to wealth, and asset allocation shows greater dependence on wealth and now falls measurably over time . Consumption and asset allocation grow faster when wealth is small and then gradually level off as wealth becomes large. And comparing scales, we see that consumption is generally lower than for $\eta=0.01$, particularly for advanced ages.

\begin{figure}[H]\centering
\begin{minipage}[b]{\textwidth} 
\includegraphics[scale=.34]{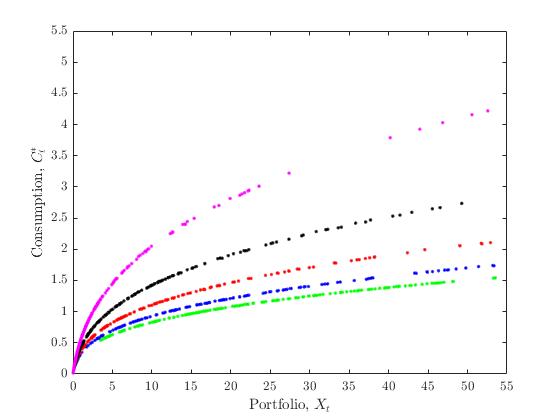}
\includegraphics[scale=.34]{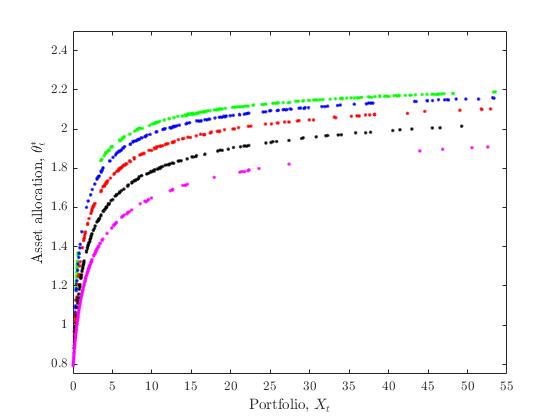}
\end{minipage}
\caption{Greedy solution with $\eta=1$, $\pi=0$ and habit $H_t=1,$ for multiple time moments.}\label{f4}
\end{figure}

Figure \ref{f4} shows the case $\eta=1$, corresponding to habit adapting very rapidly to changes in consumption. Both consumption and asset allocation now rise very rapidly when wealth is small, before levelling off. For large wealth, asset allocation is almost flat, though highly leveraged ($\theta_t>1$). We also see a continuation of the pattern in which rising $\eta$ lowers consumption, but raises asset allocation. 

\subsection{With pension income $\pi$} 
\label{NR_with_pension}
We now incorporate pension income $\pi>0$, and examine how that affects the behaviour of the greedy solution. 

We will focus on the case $\eta=1$, where we see the clearest impact, but the same effects could also be seen at lower levels of $\eta$. We will show two cases, a lower pension of $\pi=0.5$ (Figure \ref{f10_eta1_a}), and a higher pension of $\pi=1.5$ (Figure \ref{f10_eta1_c}). 

With $\pi=0.5$ the shape of the consumption curves is similar to the no-pension case, except that consumption no longer drops to 0 when wealth if very small. Quantitatively, pension naturally acts to raise consumption. Asset allocation flattens off when wealth is large, though at a higher level than when $\eta=0$. And the existence of stable pension income now allows asset allocation to spike up when wealth is small. It is natural that the impact of pension is felt most strongly at low wealth. 

For the larger pension $\pi=1.5$, the above effects are even more accentuated. But we also start to see the impact of imposing the constraint in our problem, that consumption must be $\ge \pi$. In particular, consumption is flat at the pension level when wealth is small. In other words, at small wealth the solution shows no consumption from the wealth pool. Presumably, if we had not imposed this constraint, consumption would have dipped below the pension level for low wealth (in other words, some pension income would have been invested). We believe this is also the source of the more complicated behaviour of the crossing asset allocation curves at low wealth. 

\begin{figure}[H]\centering
\begin{minipage}[b]{\textwidth} 
\includegraphics[scale=.25]{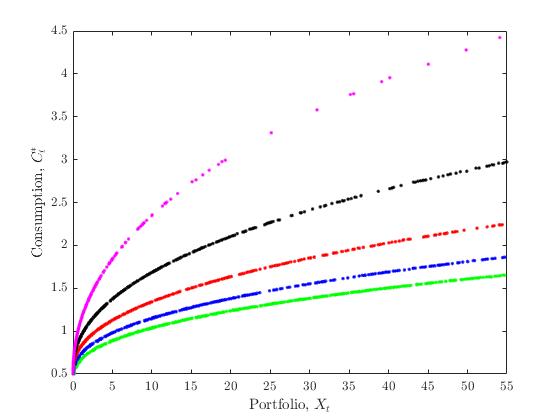}
\includegraphics[scale=.25]{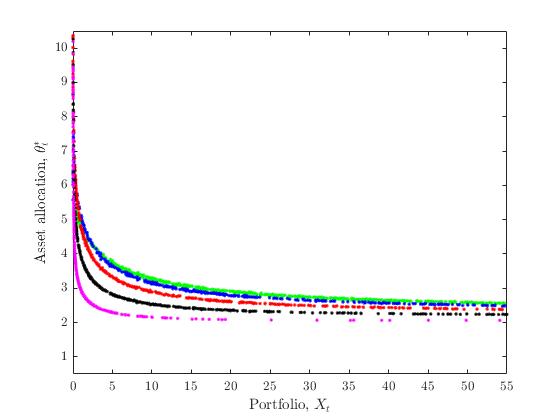}
\centering \subcaption{$\pi=0.5$}
\label{f10_eta1_a}
\end{minipage}
\begin{minipage}[b]{\textwidth} 
\includegraphics[scale=.25]{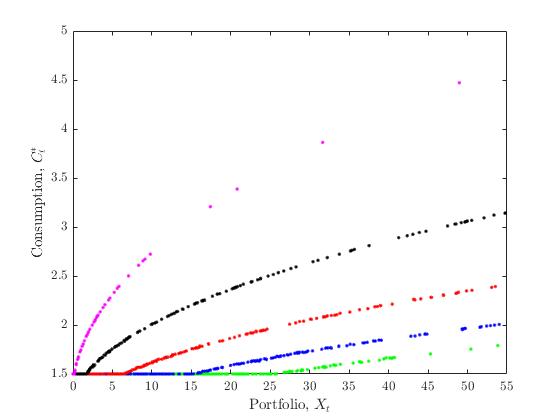}
\includegraphics[scale=.25]{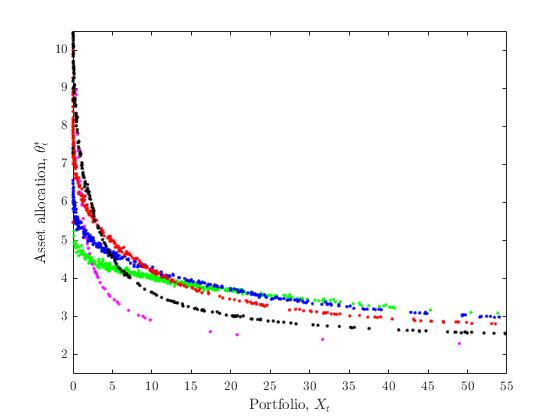}
\centering  \subcaption{$\pi=1.5$}
\label{f10_eta1_c}
\end{minipage}
\caption{Greedy solution with $\eta=1$ and habit $H_t=1,$ for multiple time moments.}
\label{f10_eta1}
\end{figure}
%
\section{Comparison of the Greedy and Optimal solutions.}
\label{comp_P_M}
Having investigated the behaviour of the greedy solution, we now wish to understand when it will be a good approximation to the global optimum (which is difficult to compute using martingale methods / duality). 

We would expect it to be a good approximation when $\eta$ is small (ie when Merton is itself a reasonable approximation), and a poor approximation when $\eta$ is large. So the real question is whether it is useful when $\eta$ is moderate (and when the Merton approximation works poorly)? We will see that it is, at least when wealth is not too large. 

We can compute the global solution using value function / pde methods, though this is taxing when $\pi\neq 0$ due to the lack of scale invariance. But when $\pi=0$ one can use scaling to reduce the dimension of the pde's, so we will carry out the comparison in that case. We refer to an earlier paper \cite{KHS} for the methodology. That paper's main focus was a model with fixed rather than variable asset allocation, but the case of variable $\theta_t$ and $\pi=0$ was also treated there. 

As before, we will consider three choices for $\eta$, namely $\eta=0.01$ (small), $\eta=0.1$ (moderate), and $\eta=1$ (large). 

Figure \ref{c_pde_maf1} shows the case $\eta=0.01$, with the optimal (pde) solution shown as solid curves, and the greedy (Monte Carlo) solution shown as dots. Again, different colours represent different moments of time. 

As expected, the agreement between the two solutions is very good, with the curves fitting the dots quite closely. 

\begin{figure}[H]\centering
\hspace{-2.0cm}
\begin{minipage}[b]{.22\textwidth} 
\includegraphics[scale=.212]{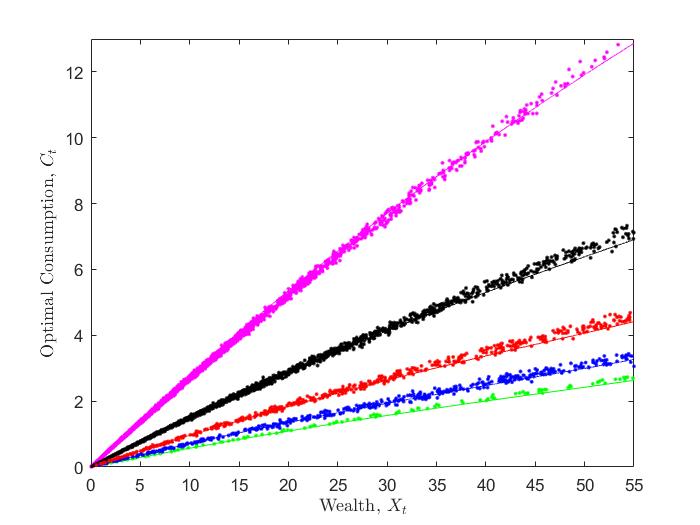}
\label{c_pde_maf1_c}
\end{minipage}
\caption{Comparison between the optimal (solid) and greedy (dots) solutions, for $\eta=0.01$, $\pi=0$, 
normalized by $H_t=1$, at differing times.}
\label{c_pde_maf1}
\end{figure}

Figure \ref{c_pde_maf234} shows the case of moderate smoothing, $\eta=0.1$; Here we saw that the Merton approximation was poor, but it is evident that our greedy approximation is still quite good, at least for wealth between 0 and 10. As wealth grows, there is some divergence between the solutions, with the greedy solution providing more aggressive consumption than optimal. 

This divergence would continue if we projected to greater wealth, which is not unreasonable, since large wealth induces much larger consumption than the habit $H_t=1$ that underlies these pictures. Our conclusion is that for moderate $\eta$ there is quite good agreement between the optimum and our greedy approximation, at least when there is not too great a mismatch between consumption and habit. But that the approximation is at best fair, once consumption and habit differ significantly. 

\begin{figure}[H]\centering
\hspace{-2.0cm}
\begin{minipage}[b]{.27\textwidth} 
\includegraphics[scale=.27]{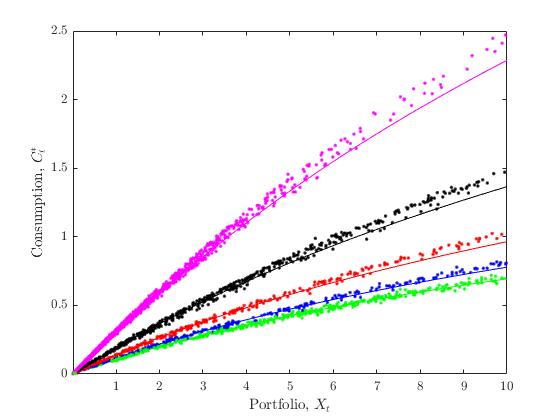}
\label{c_pde_maf234_c}
\end{minipage}
\caption{Comparison between the optimal (solid) and greedy (dots) solutions, for $\eta=0.1$, $\pi=0$, 
normalized by $H_t=1$, at differing times.}
\label{c_pde_maf234}
\end{figure}

Finally, Figure \ref{c_pde_maf5} shows the case $\eta=1$. Here there is only reasonable agreement between the solutions when wealth is quite small. So even when consumption is close to habit (here $\approx 1$), the approximation is not particularly accurate. 

%
\begin{figure}[H]\centering
\hspace{-2.0cm}
\begin{minipage}[b]{.27\textwidth} 
\includegraphics[scale=.265]{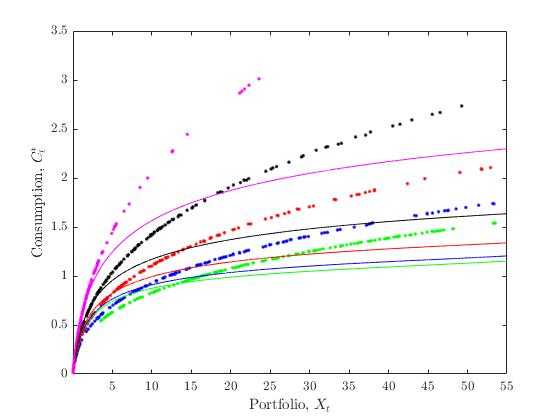}
\label{c_pde_maf5_c}
\end{minipage}
\caption{Comparison between the optimal (solid) and greedy (dots) solutions, for $\eta=1$, $\pi=0$, 
normalized by $H_t=1$, at differing times.}
\label{c_pde_maf5}
\end{figure}
%
\section{Simulations over the lifetime}
\label{WD}
The functional results presented above do not necessarily give one a sense of how consumption and wealth will evolve over the lifecycle. So here we take a simulation of the stock, and show  what that implies for the various quantities over time.

We use the greedy strategy, with the basic investment and mortality parameters used earlier. We will take $\eta=0.1$ so habit has an impact, yet the greedy strategy is close to optimal. We will look at multiple choices for pension level $\pi$, namely $\pi=0$ (black), $\pi=0.5$ (red), $\pi=1$ (blue), $\pi=1.5$ (green) and $\pi=2$ (magenta). 

We will do this for two scenarios, the first with initial habit $\bar c=1$ and initial wealth $v=10$. The second will take $\bar c=5$ and $v=30$. 

Figure \ref{wd_f21} shows consumption, wealth, and asset allocation in the case $\bar c=1$ and $v=10$. We see flat spots in consumption, due to our constraint that consumption cannot fall below pension. Initial consumption from wealth is actually higher for lower pensions, especially when pension exceeds the initial habit. But a higher pension also draws habit up faster, which creates additional consumption from wealth. So in fact, we see higher pension causing more rapid depletion of wealth. For non-zero pension, once wealth is close to exhausted, we also see asset allocation getting large, supported by that steady pension income. 

\begin{figure}[H]\centering
\hspace{-2cm}
\begin{minipage}[b]{.27\textwidth}
\includegraphics[scale=.27]{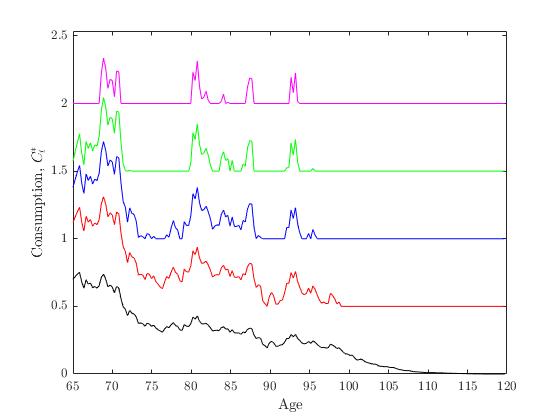} 
\subcaption{}
\label{wd_f21_a}
\end{minipage}
\hspace{0.9cm}
\begin{minipage}[b]{.27\textwidth}
\includegraphics[scale=.27]{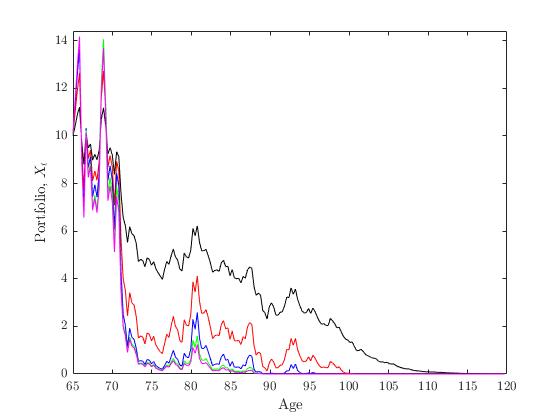} 
\subcaption{}
\label{wd_f21_b}
\end{minipage}
\hspace{0.9cm}
\begin{minipage}[b]{.27\textwidth}
\includegraphics[scale=.27]{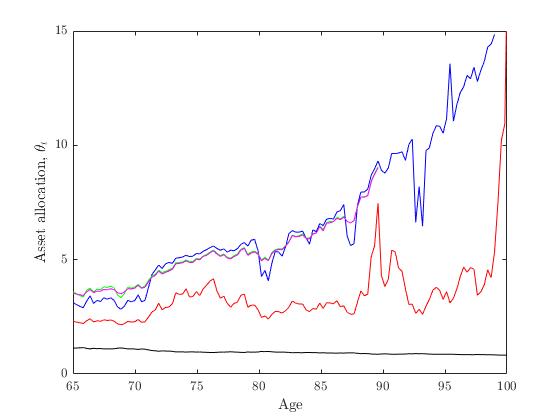} 
\subcaption{}
\label{wd_f21_c}
\end{minipage}
\caption{Lifetime simulations of the greedy strategy with $\eta=0.1$, initial habit $\bar c=1$ and initial wealth $v=10$. }
\label{wd_f21}
\end{figure}

Figure \ref{wd__f22} shows the second scenario, with $\bar c=5$ and $v=30$. Higher initial wealth leads to higher initial consumption, in comparison to the earlier scenario. Also, since initial habit is high, the dynamics will cause habit to drift down over time, which permits consumption to come down from its initially high levels. The combination leads in most cases to wealth persisting for longer than in the earlier scenario. 

\begin{figure}[H]\centering
\hspace{-2cm}
\begin{minipage}[b]{.27\textwidth}
\includegraphics[scale=.27]{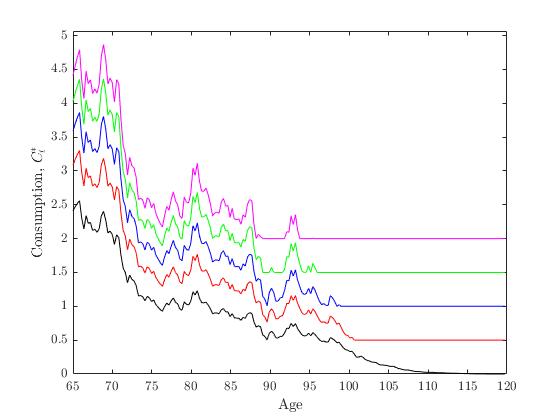} 
\subcaption{}
\label{wd__f22_1}
\end{minipage}
\hspace{0.9cm}
\begin{minipage}[b]{.27\textwidth}
\includegraphics[scale=.27]{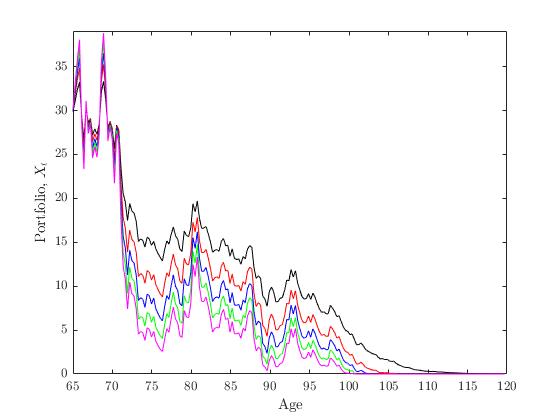} 
\subcaption{}
\label{wd__f22_2}
\end{minipage}
\hspace{0.9cm}
\begin{minipage}[b]{.27\textwidth}
\includegraphics[scale=.27]{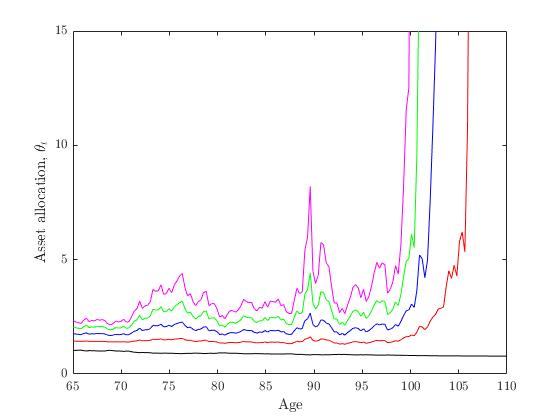} 
\subcaption{}
\label{wd__f22_3}
\end{minipage}
\caption{Lifetime simulations of the greedy strategy with $\eta=0.1$, initial habit $\bar c=5$ and initial wealth $v=30$. }
\label{wd__f22}
\end{figure}

\section{Conclusions}
We have applied martingale methods to a natural but relatively unstudied formulation of the optimal consumption problem with habit formation, in the presence of lifecycle mortality. We find a computationally tractable solution that is approximately optimal when wealth is moderate, and when the rate of reversion of habit to consumption is also moderate. Whereas the classical Merton approximation only works when this reversion rate is small. We then study the dependence of this solution on pension, among other factors. This advances the habit formation literature in the direction of coping with more realistic formulations of utility. 


\end{document}